\documentclass[12pt]{article}
    \usepackage[margin=1in]{geometry}
    \usepackage[utf8]{inputenc}
    \usepackage[T1]{fontenc}    
    \usepackage{amsmath,amssymb,amsfonts,amsthm}
    \usepackage{mathtools}
    \usepackage{physics}
    \usepackage[symbol]{footmisc}
    \usepackage{url}
    \usepackage{bigints}
    \usepackage{graphicx}
    \usepackage[title]{appendix}    
    \usepackage{tabularx}
    \usepackage{xcolor}
    \usepackage{todonotes}
\newtheorem{theorem}{Theorem}
\theoremstyle{definition}

\newtheorem{proposition}{Proposition}
\theoremstyle{remark}
\newtheorem{remark}{Remark}
\graphicspath{{./Fig/}}

\begin{document}

% The title page
\begin{titlepage}
\begin{center}
    %\vspace{0.25cm}

    \Large\textbf{HeMiTo-dynamics: a characterisation of mammalian prion toxicity using non-dimensionalisation, linear stability and perturbation analyses}
       
    \normalsize
        
    %\vspace{0.25cm}
    \setcounter{footnote}{0}
    \setlength{\footnotemargin}{0.8em}
    {\normalsize Johannes G. Borgqvist\footnote{Corresponding author: \url{johborgq@chalmers.se}}\footnote{Mathematical Sciences, Chalmers University of Technology, Gothenburg, Sweden} and Christoffer Gretarsson Alexandersen\footnote{Mathematical Institute, University of Oxford, United Kingdom}}
    \setlength{\footnotemargin}{1.8em}
        
    %\vspace{0.25cm}
        
\abstract{Prion-like proteins play crucial parts in biological processes in organisms ranging from yeast to humans. For instance, many neurodegenerative diseases are believed to be caused by the production of prion-like proteins in neural tissue. As such, understanding the dynamics of prion-like protein production is a vital step toward treating neurodegenerative disease. Mathematical models of prion-like protein dynamics show great promise as a tool for predicting disease trajectories and devising better treatment strategies for prion-related diseases. Herein, we investigate a generic model for prion-like dynamics consisting of a class of non-linear ordinary differential equations (ODEs), establishing constraints through a linear stability analysis that enforce the expected properties of mammalian prion-like toxicity. Furthermore, we identify that prion toxicity evolves through three distinct phases for which we provide analytical descriptions using perturbation analyses. Specifically, prion-toxicity is initially characterised by the healthy phase, where the dynamics are dominated by the healthy form of prions, thereafter the system enters the mixed phase, where both healthy and toxic prions interact, and lastly, the system enters the toxic phase, where toxic prions dominate, and we refer to these phases as HeMiTo-dynamics. These findings hold the potential to aid researchers in developing precise mathematical models for prion-like dynamics, enabling them to better understand underlying mechanisms and devise effective treatments for prion-related diseases.     }
    
    %\vspace{0.1cm}
    
\end{center}
\end{titlepage}

%==========================================================
% Introduction
%==========================================================
\section{Introduction}\label{sec:intro}

Prions are a class of proteins that are responsible for diseases such as Creutzfeldt-Jakob disease, kuru, and bovine spongiform encephalopathy~\cite{fraser_prions_2014, prusiner_molecular_1991}. Prions typically consist of several structural variants, some of which are healthy and some of which are harmful. The harmful, also called misfolded, prion variants can bind to healthy variants, misfolding them, and effectively converting them into toxic, misfolded protein. Misfolded prion proteins also form aggregates which eventually induce cell death. Moreover, many other proteins have been shown to have prion-like features such as having a toxic variant of itself capable of converting healthy protein into toxic protein. These prion-like proteins are believed to underlie neurodegenerative diseases in mammals~\cite{dugger_pathology_2017, brundin_prion-like_2010}. A notable example is that of amyloid beta playing a critical role in Alzheimer's disease~\cite{condello2020prion, dugger_pathology_2017}.  Longitudinal studies tracking biomarkers of Alzheimer's disease over time indicate that the toxic form increases like a sigmoid curve over time~\cite{jack2013tracking}. However, due to the long time scales of human neurodegeneration, it is difficult to investigate the biological mechanisms underlying prion dynamics in human brains. Prion-like proteins are also studied in fungi in general and baker's yeast in particular~\cite{sindi2009prion}, where such proteins often participate in cellular processes beneficial to the organism~\cite{cascarina_yeast_2014}. Prion-like protein dynamics in yeast operate on much shorter time scales and may thus be studied in greater detail. A powerful tool for inferring biological mechanisms underlying prion dynamics is mechanistic modelling, and hitherto, numerous mathematical models based on the prion hypothesis have been constructed and analysed~\cite{WEICKENMEIER2019_prion-model, Sindi17, thompson2020protein, thompson_role_2021, meisl_molecular_2020, greer_mathematical_2006} which, for instance, capture the sigmoid-like accumulation of toxic prions over time.

Prion-like dynamics can be described in different levels of detail. Note that we will refer to the non-aggregating protein variant as healthy and the aggregating variant as toxic. However, keep in mind that the aggregating variant could potentially be beneficial as opposed to pathological, in particular when discussing fungal prions. When healthy prion-like proteins are transformed into their toxic counterpart, they form aggregates of various sizes, which can be modelled explicitly~\cite{WEICKENMEIER2019_prion-model, greer_mathematical_2006}, describing the concentration evolution of each aggregate size. However, these models of aggregation are infinite-dimensional and depend on a large number of parameters, which makes numerical simulations and data fitting quite challenging.
The heterodimer model describes a simplification of prion-like replication, where aggregates are ignored. Instead, one only considers a concentration of healthy proteins and a concentration of toxic proteins, for which a reaction is transforming healthy into toxic protein~\cite{sindi2009prion, WEICKENMEIER2019_prion-model}. 
In the context of prion dynamics in the baker's yeast \textit{Saccharomyces cerevisiae}, a class of heterodimer models encapsulating other previously analysed models of mammalian prion toxicity have been recently proposed by Lemarre et al.~\cite{lemarre2020unifying}. Denoting the concentration of the healthy form by $H(t)$ and the concentration of the toxic form by $T(t)$ where the independent variable $t$ corresponds to time, we consider a heterodimer model with an arbitrary \textit{conversion function} $f$  given by
\begin{align}
    \dot{H}=&k_{1}-k_{2}H-k_{3}HTf(T)\,,\label{eq:ODE_H}\\
    \dot{T}=&k_{3}HTf(T)-k_{4}T\,,\label{eq:ODE_T}\\
    &H(t=0)=H_{0}\,,\quad{T}(t=0)=T_{0}\,,\label{eq:IC_H_and_T}
\end{align}
where time derivatives are denoted by dots, e.g. $\dot{H}=\mathrm{d}H/\mathrm{d}t$. Here, we assume that the concentrations of the two species $H$ and $T$ are measured in [nM], that time $t$ is measured in [years], and that the initial conditions in Eq. \eqref{eq:IC_H_and_T} are defined by positive constants $H_{0},T_{0}>0$. Then $k_{1}\,[\mathrm{nM}/\mathrm{year}]$ is the constant formation rate of the healthy species $H$, $k_{2}$ and $k_{4}$ are first order degradation rates of $H$ and $T$ measured in $[\mathrm{year}^{-1}]$, $k_{3}\,[\mathrm{nM}^{-1}\mathrm{year}^{-1}]$ is the second order conversion rate from the healthy to the toxic form and $f(T)$ is the conversion function which is a unitless bounded analytical differentiable function affecting the conversion rate. Mathematically speaking, both in the context of fungi and mammals, the choice of the conversion function $f$ must support the existence of two steady states, namely a so-called \textit{healthy steady state} (HSS) which is free of toxic prions and a so called \textit{toxic steady state} (TSS) characterised by toxic prions. Given the existence of two such steady states, the desired dynamical properties of prion models differ slightly for mammals and fungi.

When it comes to prion dynamics in \textit{Saccharomyces cerevisiae}, a desired mathematical property is that of \textit{bistability}~\cite{lemarre2020unifying}. Here, prions are inherited by daughter cells from mother cells after cell division, and it is the initial concentration of prion aggregates that determines whether the daughter cell will be functional or dysfunctional. Mathematically, this feature is captured by bistability implying that both the HSS and the TSS are stable simultaneously. In this case, it is the amount of prion aggregates that the daughter inherits after cell division corresponding to $T_{0}$ in Eq. \eqref{eq:IC_H_and_T} that determines whether the system evolves to the HSS or the TSS. A concrete example of a bistable model investigated by Lemarre et al.~\cite{lemarre2020unifying} is characterised by choosing the conversion function to a Hill function of the type $f(T)=K_{1}T^{n-1}/(K_{2}+T^{n})$ for some positive integer $n$ and some positive constants $K_{1},K_{2}$. 

For prion models of human and mammal neurodegeneration, the concentration of the toxic form of prions is initially low and thereafter it increases over time like a sigmoid curve~\cite{jack2013tracking}. Importantly, since the initial concentration of the toxic form $T_{0}$ in Eq. \eqref{eq:IC_H_and_T} is low, conditions on the parameters and the conversion function $f$ are typically enforced in order to capture the sigmoidal accumulation of the toxic form. Technically, these conditions result in stability critiera such that the HSS is a saddle point and the TSS is a stable node. A concrete example of such a prion model describing the dynamics of the biomarkers amyloid beta and tau proteins in the context of Alzheimer's disease was proposed by Thompson et al.~\cite{thompson2020protein} and it corresponds to the conversion function $f(T)=1$. Note that the initial concentrations of toxic prions in mammals are not a result of skewed protein distributions during cell replication, but rather variations in the propensity of brain regions to form prion-like pathology.

Nevertheless, given this latter prion model of human neurodegeneration, there are two fundamental unanswered questions. What choices of conversion function characterise models of mammalian prion toxicity where the HSS is a saddle point and the TSS is a stable node? Also, these stability criteria are ensured by means of linear stability analysis yielding knowledge about the long term dynamics of the models at hand. However, another unknown question is what features characterise the early dynamics of this type of model starting from low initial concentrations $H_{0}$ and $T_{0}$ in Eq. \eqref{eq:IC_H_and_T}? Critically, a linear stability analysis cannot provide such detailed descriptions of the early dynamics, and typically detailed knowledge about dynamical properties of ODEs at different time scales can be obtained by means of a perturbation analysis~\cite{gerlee2022weak}. Given an appropriate non-dimensionalisation of the model of interest resulting in a small perturbation parameter $\varepsilon\ll{1}$, solutions are expressed as series expansions in this perturbation parameter where the different orders in the series expansions approximate the dynamics of the system of interest at certain time scales. In particular, the initial dynamics are captured by so called \textit{outer solutions} corresponding to the $\mathcal{O}(1)$ terms in the perturbation series, and typically the original system of ODEs can be expressed as simpler ODEs that are valid on specific time scales encoded by orders of the perturbation parameter $\varepsilon$.  

In this work, we derive a class of prion models of mammalian prion toxicity. Using a non-dimensionalisation resulting in a class of models defined by a perturbation parameter corresponding to a dimensionless conversion rate, we define conditions on the conversion function $f$ such that the HSS is a saddle point and the TSS is a stable node. Using numerical simulations starting from low initial concentrations of prions, we show that the dynamics of prion models in this class are divided into three phases referred to as \textit{HeMiTo}-dynamics. First, the system goes through the \textit{healthy phase} called the He-phase, where the concentration of the healthy form normalises while that of the toxic form does not change. Thereafter, the system goes through the \textit{mixed phase} called the Mi-phase where the healthy form reaches a maximum value and the toxic form increases. Lastly, the system goes through the \textit{toxic phase} called the To-phase where the toxic form dominates the dynamics and specifically the system evolves towards the TSS. In many scenarios, it would be reasonable to assume that the healthy protein form has already normalised to regular physiological levels. In this case, the He-phase has already passed, and the introduction of a seed or alteration in system parameters introduces the Me-phase. 
Moreover, using perturbation analyses, we derive approximations of analytical solutions in the initial He- and  Mi-phases capturing the exact time dependence of solutions of our class of prion models and we validate these approximations by fitting them to numerical solutions. Lastly, we show that the choice of conversion function $f$ results in two main types of dynamics during the Mi-phase where the concentration profile for the healthy form $u(\tau)$ is either a concave function which reaches a clear maximum value similar to the epidemiological SIR model~\cite{kermack1927contribution} or it behaves more like a logistic growth function evolving towards a carrying capacity.

%==========================================================
% A general class of minimal bi-stable mechanistic prion
% models
%==========================================================
\section{Non-dimensionalisation of the class prion models yields a perturbation parameter corresponding to a conversion rate which defines time scales for the dynamics}
We conduct a non-dimensionalisation of the class of prion models in order to derive an appropriate perturbation parameter $\varepsilon$. This perturbation parameter should be small implying that $\varepsilon<1$ in the dimensionless setting, and we will subsequently use this as a basis for a perturbation analysis in order to establish distinct phases during the accumulation of toxic prions. To find such a perturbation parameter, we consider the steady states of the class of prion models of interest. 

Prion models in Eqs. \eqref{eq:ODE_H} to \eqref{eq:ODE_T} are defined by conversion functions $f$ that sustain two steady states. Technically, these are non-negative coordinates $(H^{\star},T^{\star})$ in the $(H,T)$ phase plane for which the derivatives $\dot{H}$ and $\dot{T}$ are zero. Starting with the second ODE for $T$ in Eq. \eqref{eq:ODE_T}, these steady state coordinates solve
\begin{equation}
    T^{\star}\left(k_{3}H^{\star}f(T^{\star})-k_{4}\right)=0\,.
    \label{eq:SS_2}
\end{equation}
One solution is given by $T_{1}^{\star}=0$, and substituting this value into the ODE for $H$ in Eq. \eqref{eq:ODE_H} yields the so called \textit{healthy steady state} (HSS)
\begin{equation}
    \mathrm{HSS}=\left(H^{\star},T^{\star}\right)=\left(\frac{k_{1}}{k_{2}},0\right)\,,
    \label{eq:HSS_dim}
\end{equation}
which is free of toxic prions. Moreover, we refer to the second steady state which is characterised by toxic prions as the \textit{toxic steady state} (TSS). In order for a TSS to exist, we must have that $f(T_{2}^{\star})>0$ and in this case Eq. \eqref{eq:SS_2} yields that the first coordinate of the TSS is given by $H_{2}^{\star}=k_{4}/(k_{3}f(T_{2}^{\star}))$. By substituting this value into Eq. \eqref{eq:ODE_H}, the TSS is given by
\begin{equation}
    \mathrm{TSS}=\left(H_{2}^{\star},T_{2}^{\star},\right)=\left(\frac{k_{4}}{k_{3}f(T_ {2}^{\star})},\frac{k_{1}}{k_{4}}\left(1-\frac{k_{2}k_{4}}{k_{1}k_{3}f(T_{2}^{\star})}\right)\right)\,.
    \label{eq:TSS_dim}
\end{equation}
Of particular interest is the second coordinate $T_{2}^{\star}$, and specifically we consider the following dimensionless parameter
\begin{equation}
    \varepsilon=\frac{k_{2}k_{4}}{k_{1}k_{3}f(T_{2}^{\star})}\,.
    \label{eq:eps}
\end{equation}
In order for a biologically-reasonable TSS to exist, we require that both $H_{2}^{\star}$ and $T_{2}^{\star}$ are positive. Accordingly, the conversion function must satisfy $f(T_{2}^{\star})>0$ so that $H_{2}^{\star}>0$, and $f(T_{2}^{\star})$ must be chosen so that $\varepsilon$ in Eq. \eqref{eq:eps} lies in the interval $\varepsilon\in(0,1)$ which ensures that $T_{2}^{\star}>0$ according to Eq.~\ref{eq:TSS_dim}. Consequently, $\varepsilon$ is our perturbation parameter which we subsequently use as the basis for our non-dimensionalisation and perturbation analysis. In light of this perturbation parameter, we introduce the following dimensionless time variable
\begin{equation}
    \tau=\left(\frac{k_{1}k_{3}f(T_{2}^{\star})}{k_{2}}\right)t\,,
    \label{eq:tau}
\end{equation}
together with the following dimensionless states
\begin{equation}
    u=\left(\frac{k_{3}}{k_{4}}\right)H\,,\quad{v}=\left(\frac{k_{3}}{k_{4}}\right)T\,,
    \label{eq:u_and_v}
\end{equation}
and the following dimensionless parameters
\begin{equation}
   c_{1}=\frac{k_{2}}{k_{4}f(T_{2}^{\star})}\,,\quad{c}_{2}=\frac{k_{2}^{2}}{k_{1}k_{3}f(T_{2}^{\star})}\,.
    \label{eq:parameters}
\end{equation}
Expressing the original class of prion models in terms of these dimensionless variables and parameters (for details, see Appendix \ref{sec:non_dim}) yields the following equivalent dimensionless system
\begin{align}    
\dot{u}=&c_{1}-c_{2}u-\varepsilon{u}vf(v)\,,\label{eq:ODE_u}\\
\dot{v}=&\varepsilon\left(uvf(v)-v\right)\,,\label{eq:ODE_v}\\
&u(\tau=0)=u_{0}\,,\quad{v(\tau=0)}=v_{0}\label{eq:IC_u_and_v}\,,
\end{align}
where, again, time derivatives are denoted by dots, e.g. $\dot{u}=\mathrm{d}u/\mathrm{d}\tau$. Given this dimensionless system, we proceed by conducting a linear stability analysis in order to classify a subclass of models of mammalian prion toxicity defined by choices of conversion functions $f$ having the same dynamical properties as the reference model by Thompson et al.~\cite{thompson2020protein} defined by $f(v)=1$.

%==========================================================
% A class of models describing prion toxicity
%==========================================================
\section{Defining a class of models of mammalian prion toxicity by means of linear stability analysis}
Dynamics describing \textit{mammalian prion toxicity} implying the accumulation of toxic prions over time starting from low initial concentrations of prions are based on two mathematical properties. First, as mentioned previously, the existence of a \textit{healthy steady state} (HSS) which is free of toxic prions and a \textit{toxic steady state} (TSS) characterised by toxic prions. Second, prion toxicity corresponds to dynamics where the HSS is a saddle point and the TSS is a stable node. Here, we derive conditions on the arbitrary conversion function $f$ allowing for prion toxicity, and initially we present necessary conditions for the existence of two steady states (Theorem \ref{thm:SS}).

\begin{theorem}[Existence of healthy and toxic steady states]
The system of ODEs in Eqs. \eqref{eq:ODE_u} and \eqref{eq:ODE_v} has a \textit{healthy steady state} (HSS) given by
\begin{equation}
    \mathrm{HSS}=\left(u_{1}^{\star},v_{1}^{\star}\right)=\left(\frac{c_{1}}{c_{2}},0\right)\,.
    \label{eq:HSS}
\end{equation}
If there exists a $v_{2}^{\star}>0$ such that $f(v_{2}^{\star})>0$ that solves
\begin{equation}
    v_{2}^{\star}=\frac{1}{\varepsilon}\left(c_{1}-\frac{c_{2}}{f(v_{2}^{\star})}\right)\,,
    \label{eq:v_2_star}
\end{equation}
then the system has a \textit{toxic steady state} (TSS) given by
\begin{equation}
    \mathrm{TSS}=\left(u_{2}^{\star},v_{2}^{\star}\right)=\left(\frac{1}{f(v_{2}^{\star})},\frac{1}{\varepsilon}\left(c_{1}-\frac{c_{2}}{f(v_{2}^{\star})}\right)\right)\,.
    \label{eq:TSS}
\end{equation}
\label{thm:SS}
\end{theorem}
\begin{proof}
See Appendix \ref{sec:SS}.
\end{proof}
\begin{remark}
    The coordinate $v_{2}^{\star}$ of the TSS is given by the intersection (FIG. \ref{fig:TSS}) between the conversion function $f$ and the function $g$ defined by
    \begin{equation}
    g(v)=\frac{c_{2}}{c_{1}-\varepsilon{v}}\,.
        \label{eq:g}
    \end{equation}
    \label{remark:g}
\end{remark}
Provided two steady states, we characterise prion toxicity by means of linear stability analysis. To this end, we define a condition determining when the HSS is a saddle point (Proposition \ref{thm:saddle_HSS}).

\begin{figure}[htbp!]
    \centering
    \includegraphics{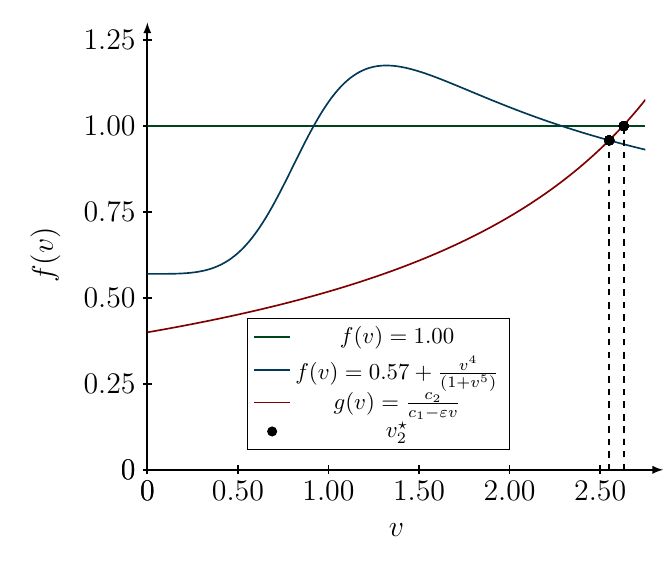}
    \caption{\textit{The toxic steady state (TSS) given by the intersection between the functions $f$ and $g$}. The coordinate $v_{2}^{\star}$ of the TSS is given by the intersection between the conversion function $f(v)$ and the function $g(v)=(1-\varepsilon)/(c(1-\varepsilon{v}))$ illustrated by the red curve. This is visualised in two cases: (\textbf{A}) $f(v)=1.00$ illustrated by the green curve which yields $\mathrm{TSS}=(u_{2}^{\star},v_{2}^{\star})=(1.00,2.63)$, and (\textbf{B}) $f(v)=0.57+(v^{4}/(1+v^{5}))$ illustrated by the blue curve which yields $\mathrm{TSS}=(u_{2}^{\star},v_{2}^{\star})=(1.04,2.55)$. The parameters defining the illustrated curves are $c_{1}=1.75$, $c_{2}=0.70$ and $\varepsilon=0.40$.}
    \label{fig:TSS}
\end{figure}

\begin{proposition}[Condition defining the HSS as saddle point]
The HSS in Eq. \eqref{eq:HSS} is a saddle point if the conversion function $f$ satisfies
\begin{equation}
    f(0)>\frac{c_{2}}{c_{1}}\,.
    \label{eq:HSS_saddle}
\end{equation}
    \label{thm:saddle_HSS}
\end{proposition}
\begin{proof}
    See Appendix \ref{sec:saddle_HSS}.
\end{proof}
\begin{remark}
    In terms of the function $g$ in Eq. \eqref{eq:g}, the inequality in Eq. \eqref{eq:HSS_saddle} determining when the HSS is a saddle point is given by $f(0)>g(0)$. This fact together with Remark \ref{remark:g} implies that the function $f$ is bounded from below by the function $g$ until their first intersection, which is expressed as follows
    \begin{equation}
        f(v)>{g(v)}\,\forall{v}\in[0,v_{2}^{\star})\,,\quad{f}(v_{2}^{\star})=g(v_{2}^{\star})\,.
        \label{eq:f_bound_by_g}
    \end{equation}
\end{remark}

In addition to the HSS being a saddle point, by studying the stability properties of the linearised system around the TSS we define conditions on the arbitrary functions $f$ resulting in dynamics characterising mammalian prion toxicity (Theorem \ref{thm:stable_TSS}). 

\begin{theorem}[Conditions defining mammalian prion toxicity]
    Prion models in Eqs. \eqref{eq:ODE_u} and \eqref{eq:ODE_v} defined by conversion functions $f$ sustaining a $\mathrm{TSS}=(u_{2}^{\star},v_{2}^{\star})$ given by Eq. \eqref{eq:TSS} as well as satisfying the condition in Eq. \eqref{eq:HSS_saddle} together with the condition
    \begin{equation}
        f'(v_{2}^{\star})\leq 0\,,
    \end{equation}
    have a HSS that is a saddle point and a TSS that is a stable node.
    \label{thm:stable_TSS}
\end{theorem}
\begin{proof}
    See Appendix \ref{sec:stable_TSS}.
\end{proof}
\begin{remark}
    The two functions $f(v)=1.00$ corresponding to the model by Thompson et al.~\cite{thompson2020protein} with $f'(v_{2}^{\star})=0$ and
    \begin{equation}
        f(v)=0.57+\frac{v^{4}}{1+v^{5}}\,,
        \label{eq:f_2}
    \end{equation}
    with $f'(v_{2}^{\star})<0$ result in models describing mammalian prion toxicity according to Theorem \ref{thm:stable_TSS} (FIG. \ref{fig:TSS}).
\end{remark}
In total, the linear stability analysis yields conditions (Theorems \ref{thm:SS} and \ref{thm:stable_TSS}) on the conversion functions $f$ in the class of prion models defined by Eqs. \eqref{eq:ODE_u} and \eqref{eq:ODE_v} ensuring an accumulation of toxic prions over time. Specifically, these conditions define the long term dynamics as solutions of our class of prion models in this case approach the TSSs. Nevertheless, these results from the linear stability analysis say nothing about the corresponding short term dynamics. Next, we provide a qualitative description of both short and long term dynamics by means of simulations.

%==========================================================
% HeMiTo-dynamics
%==========================================================
\section{HeMiTo-dynamics: prion toxicity is characterised by three phases}

The dynamics of mammalian prion toxicity characterised herein rely on two imposed conditions. First, the initial concentrations of healthy and toxic proteins are low, and second, the concentration of the toxic form increases and accumulates over time.
Mathematically, these assumptions correspond to low initial concentrations $u_{0}=\mathcal{O}(1)<\mathcal{O}(1/\varepsilon)$ and $v_{0}=\mathcal{O}(1)<\mathcal{O}(1/\varepsilon)$ in Eq. \eqref{eq:IC_u_and_v} and that parameters and conversion functions $f$ are chosen such that the HSS is a saddle point and the TSS is a stable node in accordance with Theorem \ref{thm:stable_TSS}. Additionally, we assume that the conversion function $f(v)$ and toxic concentration $v$ is of the same order $f(v) \sim \mathcal{O}(v)$, and therefore we have $f(v_{0})=\mathcal{O}(1)$ initially. Given these conditions, the dynamics characterising mammalian prion toxicity are divided into three phases captured by the acronym \textit{HeMiTo} (FIG. \ref{fig:HeMiTo}). 
First, the system goes through the \textit{healthy phase} referred to as the He-phase where the healthy form dominates the dynamics. In this phase, the healthy concentration approaches the HSS, and the concentration of the toxic form is constant (and low).
Given that the HSS is a saddle point and that the initial conditions are chosen relatively close to the HSS in the $(u,v)$ phase plane, solution trajectories in the He-phase move along the stable eigenvector of the linearised system around the HSS. Second, the system goes through the \textit{mixed phase} referred to as the Mi-phase where both species interact. Specifically, the trajectories of the healthy form attain maxima while the trajectories of the toxic form increase over time during the Mi-phase. Third, the system goes through the \textit{toxic phase} referred to as the To-phase where the toxic form dominates the dynamics. Mathematically, the dynamics of the system in the To-phase are approximately captured by the linearised system around the TSS which was previously considered in the linear stability analysis. Subsequently, we present similar approximations describing the dynamics during the He- and Mi-phases.  
It is, however, important to note that in biological settings, the initial concentration of healthy protein would not be small, but rather, close to the HSS. In this case, the healthy phase is trivial and characterised by constant concentrations of the healthy form and low concentrations of the toxic form. This point is addressed further in the Discussion.

\begin{figure}[htbp!]
    \centering
    \includegraphics{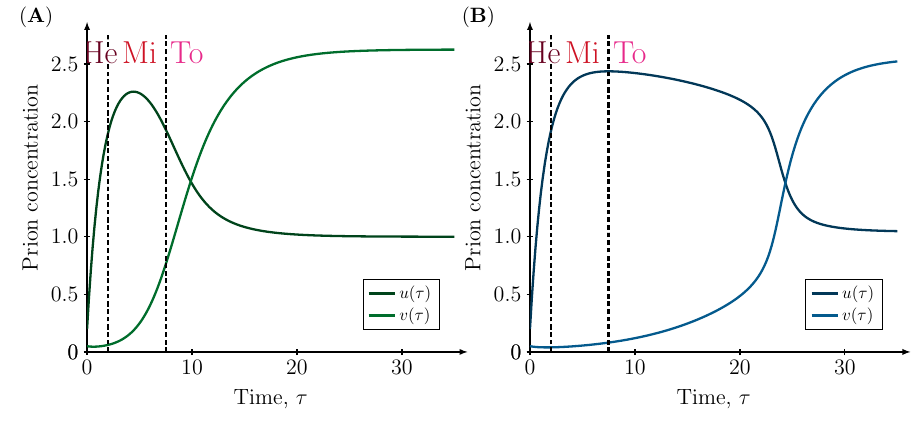}
    \caption{\textit{HeMiTo-dynamics}. The dynamics of the system of ODEs in Eqs. \eqref{eq:ODE_u} and \eqref{eq:ODE_v} characterising mammalian prion toxicity are divided into three phases remembered by the acronym \textit{HeMiTo}. First, the system goes through the \textit{healthy phase} referred to as the He-phase where the healthy form $u(\tau)$ approaches the healthy steady-state, and the toxic form $v(\tau)$ is constant. Second, the system goes through the \textit{mixed phase} referred to as the Mi-phase where the healthy and toxic forms interact. Specifically, the healthy form begins to decrease while the toxic form increases throughout the Mi-phase. Third, the system goes through the \textit{toxic phase} referred to as the To-phase where the toxic form dominates the dynamics over the healthy form. Here, solutions of the class of prion models evolve towards the stable TSS. This type of dynamics is visualised in two cases defined by distinct conversion functions: (\textbf{A}) $f(v)=1.00$, and (\textbf{B}) $f(v)=0.57+(v^{4}/(1+v^{5}))$.  The parameters defining the illustrated curves are $c_{1}=1.75$, $c_{2}=0.70$ and $\varepsilon=0.40$, and the initial conditions in both cases are given by $(u_{0},v_{0})=(0.20,0.05)$.}
    \label{fig:HeMiTo}
\end{figure}

%==========================================================
% He-phase
%==========================================================
\section{Finding approximate analytical solutions in the He-phase by means of perturbation ans\"{a}tze}
A powerful method for obtaining approximate analytical solutions of non-linear ODEs characterised by distinct phases is that of perturbation analysis~\cite{gerlee2022weak}. Given a dimensionless non-linear system of ODEs containing a small perturbation parameter $\varepsilon<1$, analytical solutions are approximated by series expansions in this perturbation parameter. Specifically, perturbation ans\"{a}tze are substituted into the original system of non-linear ODEs and then analytical solutions of simpler ODEs corresponding to the various orders of $\varepsilon$ are found. The first terms of order $\mathcal{O}(1)$ in these series expansions are referred to as \textit{outer solutions} and they describe the initial dynamics on a short time scale. Here, we define appropriate perturbation ans\"{a}tze where the outer solutions describe the dynamics in the He- and Mi-phases.

%The second terms of order $\mathcal{O}(\varepsilon)$ are referred to as \textit{inner solutions} and they describe the subsequent dynamics at a longer time scale. Typically, inner and outer solutions are matched where the initial conditions for the inner solutions equal the steady state values of the outer solutions. In this way, approximate analytical solutions that are valid in different phases of the dynamics can be obtained which result in a more detailed description of the dynamical properties of the system at hand. Here, we conduct a perturbation analysis of the prion models defined by the non-linear system of ODEs as well as initial conditions in Eqs. \eqref{eq:ODE_u} to \eqref{eq:IC_u_and_v} in order to characterise the behaviour of solution trajectories in the He- and Mi-phases. 

Initially, the dynamics are dominated by the formation of healthy prions while the conversion from healthy to toxic prions is comparatively small. Provided low initial prion concentrations in Eq. \eqref{eq:IC_u_and_v} implying that $u_{0}=\mathcal{O}(1)$ and $v_{0}=\mathcal{O}(1)$, this means that the formation rate ``$c_{1}$'' and the degradation rate ``$c_{2}u$'' in the ODE for $u$ in Eq. \eqref{eq:ODE_u} are of the same order while the conversion rate ``$\varepsilon{u}vf(v)$'' is small, and in particular we have $u(\tau)=\mathcal{O}(1)$ for early times $\tau$ close to $0$. Accordingly, consider the following regular perturbation ans\"{a}tze for the healthy and toxic form, respectively:
\begin{align}
    u(\tau)&=u_{\mathrm{He}}(\tau)+u_{1}(\tau)\varepsilon+\mathcal{O}(\varepsilon^{2})\,,\label{eq:u_approx}\\
    v(\tau)&=v_{\mathrm{He}}(\tau)+v_{1}(\tau)\varepsilon+\mathcal{O}(\varepsilon^{2})\,.\label{eq:v_approx}
\end{align}
 where the outer solutions $u_{\mathrm{He}}$ and $v_{\mathrm{He}}$ describe the dynamics in the He-phase. By substituting these ans\"{a}tze into the original class of prion models in Eqs. \eqref{eq:ODE_u} and \eqref{eq:ODE_v} and solving for the leading terms, we find approximate analytical solutions in the He-phase (Theorem \ref{thm:He}). 

\begin{theorem}[Approximate analytical solutions of prion models in the He-phase]
The outer solutions $u_{\mathrm{He}}(\tau)$ and $v_{\mathrm{He}}(\tau)$ in Eqs. \eqref{eq:u_approx} and \eqref{eq:v_approx} approximating the early dynamics when $u(\tau)=\mathcal{O}(1)$ of the system in Eqs. \eqref{eq:ODE_u} and \eqref{eq:ODE_v} under the assumption that $f(v)=\mathcal{O}(1)$ are given by
\begin{align}
    u_{\mathrm{He}}(\tau)&=\frac{c_{1}}{c_{2}}-\left(\frac{c_{1}}{c_{2}}-u_{0}\right)\exp\left(-c_{2}\tau\right)\,,\label{eq:u_He}\\
    v_{\mathrm{He}}(\tau)&=v_{0}\,,\label{eq:v_He}
\end{align}
where $u_{0}=\mathcal{O}(1)<\mathcal{O}(1/\varepsilon)$ and $v_{0}=\mathcal{O}(1)<\mathcal{O}(1/\varepsilon)$ are the initial conditions in Eq. \eqref{eq:IC_u_and_v}.
\label{thm:He}    
\end{theorem}
\begin{proof}
    See Appendix \ref{sec:He}.
\end{proof}
Importantly, these approximations agree with numerical solutions of two particular prion models described by the ODEs in Eqs. \eqref{eq:ODE_u} to \eqref{eq:IC_u_and_v} with conversion functions $f(v)=1.00$ and $f(v)=0.57+((v^{4})/(1+v^{5}))$, respectively (FIG. \ref{fig:He}). These approximations are valid while the concentration of the healthy form is low, i.e. $u(\tau)=\mathcal{O}(1)$, and when this concentration becomes sufficiently high meaning $u(\tau)=\mathcal{O}(1/\varepsilon)$ the system enters the subsequent Mi-phase.

\begin{figure}[htbp!]
    \centering
    \includegraphics{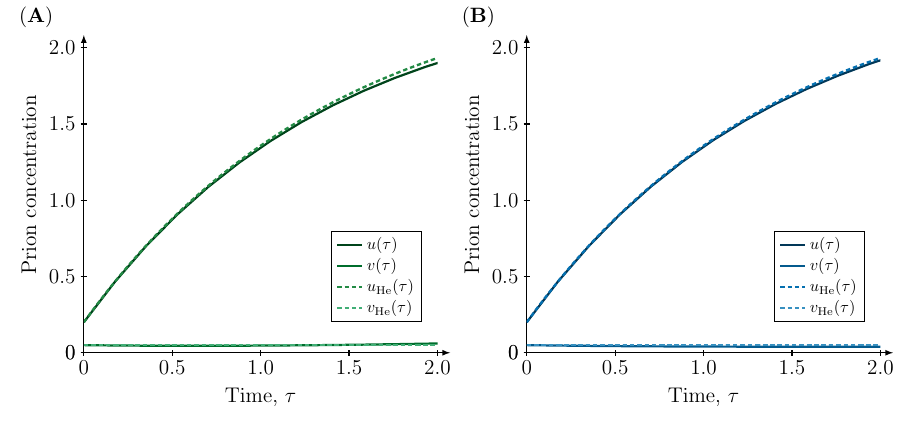}
    \caption{\textit{Prion dynamics during the He-phase}. The dynamics of the system of ODEs in Eqs. \eqref{eq:ODE_u} and \eqref{eq:ODE_v} in the early He-phase are illustrated. Specifically, numerical solutions $u(\tau)$ and $v(\tau)$ are compared to their respective approximations $u_{\mathrm{He}}(\tau)$ and $v_{\mathrm{He}}(\tau)$ in Eqs. \eqref{eq:u_He} and \eqref{eq:v_He}, respectively. This is visualised in two cases defined by distinct conversion functions: (\textbf{A}) $f(v)=1.00$, and (\textbf{B}) $f(v)=0.57+(v^{4}/(1+v^{5}))$.  The parameters defining the illustrated curves are $c_{1}=1.75$, $c_{2}=0.70$ and $\varepsilon=0.40$, and the initial conditions in both cases are given by $(u_{0},v_{0})=(0.20,0.05)$.}
    \label{fig:He}
\end{figure}

%==========================================================
% Mi-phase
%==========================================================
\section{Characterising two distinct types of dynamics during the Mi-phase by means of perturbation and asymptotic analyses}
During the Mi-phase, the conversion rate from healthy to toxic prions dominates the dynamics over the formation rate of the healthy form. In particular, since the healthy form approaches the HSS in the He-phase, we have that $u(\tau)=\mathcal{O}(1/\varepsilon)$ which implies that the constant formation rate $c_{1}=\mathcal{O}(1)<\mathcal{O}(1/\varepsilon)$ is comparatively small during the Mi-phase. 

To capture the dynamics in this phase, we re-scale the states as follows
\begin{equation}
    U(\tau)=\varepsilon{u}(\tau)\,,\quad{V}(\tau)=\varepsilon{v}(\tau)\,.
    \label{eq:states_Mi}
\end{equation}
Moreover, we approximate the conversion function by
\begin{equation}
    f(v)\approx{f}(v_{\mathrm{He}})+\mathcal{O}(\varepsilon)=f(v_{0})+\mathcal{O}(\varepsilon)\,,
    \label{eq:f_approx}
\end{equation}
where $v_{0}=\mathcal{O}(1)$ is the initial concentration of the toxic form in Eq. \eqref{eq:IC_u_and_v} and where $f(v_{0})=\mathcal{O}(1)$. This approximation can be justified by the arguments presented by Gerlee~\cite{gerlee2022weak} which state that the approximation is, in fact, exact when the conversion function $f$ is a polynomial. Also, the approximation is accurate for non-polynomial conversion functions $f$ granted continuity over a closed interval, as they, in turn, can be approximated by polynomials as guaranteed by the Stone-Weierstrass theorem~\cite{stone-weierstrass}. In light of this approximation for the conversion function $f$ in Eq. \eqref{eq:f_approx}, by multiplying the original ODEs for $u$ and $v$ in Eqs. \eqref{eq:ODE_u} and \eqref{eq:ODE_v} by $\varepsilon$, we obtain an approximate system of ODEs describing the dynamics during the Mi-phase in terms of the new states $U$ and $V$ in Eq. \eqref{eq:states_Mi} which is given by
\begin{align}
    \dot{U}&=c_{1}\varepsilon-c_{2}U-UVf(v_{0})\,,\label{eq:ODE_U}\\
    \dot{V}&=UVf(v_{0})-\varepsilon{V}\label{eq:ODE_V}\,.
\end{align}
Next, we analyse the initial dynamics of this approximate system by means of perturbation methods, and accordingly we consider the following perturbation ans\"{a}tze
    \begin{align}
    U(\tau)&=u_{\mathrm{Mi}}(\tau)+U_{1}(\tau)\varepsilon+\mathcal{O}(\varepsilon^{2})\,,\label{eq:U_approx}\\
    V(\tau)&=v_{\mathrm{Mi}}(\tau)+V_{1}(\tau)\varepsilon+\mathcal{O}(\varepsilon^{2})\,.\label{eq:V_approx}
\end{align}
Specifically, we aim at describing the approximate functional forms of the outer solutions $u_{\mathrm{Mi}}$ and $v_{\mathrm{Mi}}$ where the explicit time-dependence is captured. Substituting these ans\"{a}tze into the system of ODEs in Eqs. \eqref{eq:ODE_U} and \eqref{eq:ODE_V} and solving for the leading terms yields 
\begin{align}
    \dot{u}_{\mathrm{Mi}}&=-u_{\mathrm{Mi}}\left(c_{2}+v_{\mathrm{Mi}}f(v_{0})\right)\,,\label{eq:ODE_u_MI}\\
    \dot{v}_{\mathrm{Mi}}&=v_{\mathrm{Mi}}u_{\mathrm{Mi}}f(v_{0})\label{eq:ODE_v_Mi}\,.
\end{align}
Given this system, we find asymptotic approximations of the solutions $u_{\mathrm{Mi}}(\tau)$ and $v_{\mathrm{Mi}}(\tau)$ that explicitly describe the time dependence. Importantly, the structure of this model of prion dynamics during the Mi-phase is similar to the well-known epidemiological SIR model originally formulated by McKendrick and Kermack~\cite{kermack1927contribution}. Recently, asymptotic approximations of the solutions of the SIR model were found by integrating the ODEs of interest with respect to time and thereafter approximating the resulting unknown integrals~\cite{prodanov2023asymptotic}. Using this technique, the ODEs for $u_{\mathrm{Mi}}$ and $v_{\mathrm{Mi}}$ in Eqs. \eqref{eq:ODE_u_MI} and \eqref{eq:ODE_v_Mi}, respectively, can be formulated as equivalent integral equations
\begin{align}
    u_{\mathrm{Mi}}(\tau)&=\tilde{C}_{1}\exp\left(-c_{2}\tau\right)\exp\left(-f(v_{0})\int_{0}^{\tau}v_{\mathrm{Mi}}(s)\mathrm{d}s\right)\,,\label{eq:int_u_MI}\\
    v_{\mathrm{Mi}}(\tau)&=v_{0}\exp\left(f(v_{0})\int^{\tau}_ {0}u_{\mathrm{Mi}}(s)\mathrm{d}s\right)\label{eq:int_v_Mi}\,,
\end{align}
for some integration constant $\tilde{C}_{1}$. Starting with the last integral equation, a linearisation of the integral on the right hand side of Eq. \eqref{eq:int_v_Mi} around $\tau=0$ yields the following functional form for the evolution of the concentration of toxic prions during the Mi-phase
\begin{equation}
    v_{\mathrm{Mi}}(\tau)\approx{C}_{1}\exp(C_{2}\tau)\,,
    \label{eq:v_Mi}
\end{equation}
where $C_{1}$ and $C_{2}$ are positive constants. This indicates that the concentration of the toxic form increases exponentially during the Mi-phase. Better still, we validate this functional form by fitting such exponential functions to numerical solutions $v$ of the original ODE system in Eqs. \eqref{eq:ODE_u} and \eqref{eq:ODE_v} during the Mi-phase (FIG. \ref{fig:Mi}B and \ref{fig:Mi}D), and importantly these numerical solutions are well approximated by exponential functions. Moreover, assuming that the toxic form $v_{\mathrm{Mi}}(\tau)$ is approximately given by an exponential function allows us to approximate the functional form of the healthy form $u_{\mathrm{Mi}}(\tau)$.

\begin{figure}[htbp!]
    \centering
    \includegraphics{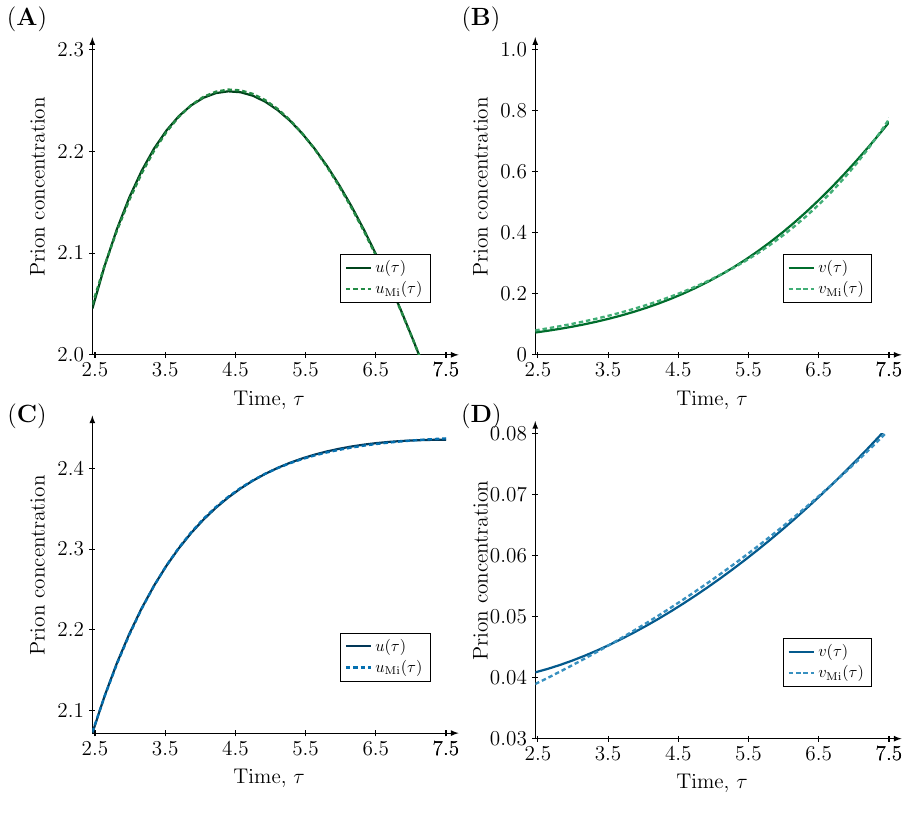}
    \caption{\textit{Prion dynamics during the Mi-phase}. The dynamics of the system of ODEs in Eqs. \eqref{eq:ODE_u} and \eqref{eq:ODE_v} in the Mi-phase are illustrated where the healthy form $u$ attains maxima and the toxic form increases exponentially. Specifically, the approximations $u_{\mathrm{Mi}}(\tau)$ and $v_{\mathrm{Mi}}(\tau)$ in Eqs. \eqref{eq:u_Mi} and \eqref{eq:v_Mi} have been fitted to numerical solutions $u(\tau)$ and $v(\tau)$ in the Mi-phase. 
    This is illustrated in two cases defined by distinct conversion functions: the top row where $f(v)=1.00$ and the bottom row where $f(v)=0.57+(v^{4}/(1+v^{5}))$. The calibrations yield the following fitted parameters (\textbf{A}) $(C_{3},C_{4},C_{5},C_{6})=(5.21,0.91,0.13,-3.16)$ in Eq. \eqref{eq:u_Mi} for $u_{\mathrm{Mi}}(\tau)$ with $u_{\max}=2.26$, (\textbf{B}) $(C_{1},C_{2})=(0.08,0.04)$ in Eq. \eqref{eq:v_Mi} for $v_{\mathrm{Mi}}(\tau)$, (\textbf{C}) $(C_{3},C_{4},C_{5},C_{6})=(-0.37,-0.13,1.08,2.44)$ in Eq. \eqref{eq:u_Mi} for $u_{\mathrm{Mi}}(\tau)$ with $u_{\max}=2.44$ and (\textbf{D}) $(C_{1},C_{2})=(0.45,0.14)$ in Eq. \eqref{eq:v_Mi} for $v_{\mathrm{Mi}}$. The parameters defining the illustrated curves are $c_{1}=1.75$, $c_{2}=0.70$ and $\varepsilon=0.40$, and the initial conditions in both cases are given by $(u_{0},v_{0})=(0.20,0.05)$.}
    \label{fig:Mi}
\end{figure}

The assumption that the toxic form increases exponentially implies that the healthy form is described by the quotient between a double exponential and an exponential function. Specifically, approximating $v_{\mathrm{Mi}}$ by an exponential function implies that the last factor on the right hand side of the integral equation for $u_{\mathrm{Mi}}$ in Eq. \eqref{eq:ODE_u_MI} is approximately given by a double exponential function of the type $\exp(-K_{1}\exp(K_{2}\tau))$ where $K_{1},K_{2}$ are constants. Such a double exponential function has the following series representation
\begin{equation}
    \exp\left(-K_{1}\exp\left(K_{2}\tau\right)\right)=\sum_{j=0}^{\infty}\frac{(-K_{1}\exp(K_{2}\tau))^{j}}{j!}\,.
    \label{eq:double_exp}
\end{equation}
Approximating the rightmost factor on the right hand side of the integral equation for $u_{\mathrm{Mi}}$ in Eq. \eqref{eq:int_u_MI} by a truncated series representation of a double exponential and linearising some of the exponential terms yields the following functional form for the evolution of the concentration of healthy prions during the Mi-phase
\begin{equation}
    u_{\mathrm{Mi}}(\tau)\approx(C_{3}+C_{4}\tau)\exp(-C_{5}\tau)+C_{6}\,,
    \label{eq:u_Mi}
\end{equation}
where $C_{3}$, $C_{4}$ and $C_{6}$ are arbitrary constants and $C_{5}$ is a positive constant. For realistic values of these constants, the concentration profile $u_{\mathrm{Mi}}(\tau)$ is positive throughout the time interval for which the above approximation is valid. In words, this approximation is given by the quotient between a linear and an exponential function plus a constant and its mathematical properties allows it to capture the characteristic feature of the concentration profile of the healthy form during the Mi-phase, namely that it attains a maximum value $u_{\max}$. Specifically, this maximum value is allowed by the fact that the linear function increases faster than the exponential function initially while the converse is true at later points in time. In fact, our approximation describes two interesting and qualitatively distinct cases of dynamics during the Mi-phase.

 In the first case, the approximation for $u_{\mathrm{Mi}}(\tau)$ is a concave function of time reminiscent of the epidemiological SIR model by McKendrick and Kermack~\cite{kermack1927contribution}. Under these circumstances, the linear constants $C_{3},C_{4}$ are positive and then the maximum value $u_{\max}$ is attained at time $\tau=\tau_{\max}$ defined by $\left.\dot{u}_{\mathrm{Mi}}\right|_{\tau=\tau_{\max}}=0$. In the second case, $u_{\mathrm{Mi}}$ behaves like a simple Verhulst model of population dynamics~\cite{verhulst1838logistic} meaning a logistic growth model. Technically, the curve $u_{\mathrm{Mi}}(\tau)$ approaches the carrying capacity $C_{6}$ which is a positive constant while the linear constants $C_{3},C_{4}$ are negative. Introducing new positive linear constants $\tilde{C}_{3}=-C_{3}$ and $\tilde{C}_{4}=-C_{4}$, we have that $u_{\mathrm{Mi}}(\tau)\approx{C}_{6}-(\tilde{C}_{3}+\tilde{C}_{4}\tau)\exp(-C_{5}\tau)$ and clearly the quotient between the linear and exponential functions decreases over time. Consequently, the maximum value of $u_{\mathrm{Mi}}$ in this case is given by $u_{\max}\approx\underset{{\tau\rightarrow+\infty}}{\lim}u_{\mathrm{Mi}}(\tau)\approx{C}_{6}$. To summarise both of these cases in one equation, the maximum concentration of the healthy form $u_{\max}$ is approximately given by
\begin{equation}
    u_{\max}\approx\max\left(\left[u_{\mathrm{Mi}}\left(\tau=\frac{C_{4}-C_{3}C_{5}}{C_{4}C_{5}}\right)\right],C_{6}\right)=\max\left(\left[\frac{C_{4}}{C_{5}}\exp\left(-\left(1-\frac{C_{3}C_{5}}{C_{4}}\right)\right)+C_{6}\right],C_{6}\right)\,.
    \label{eq:u_max}
\end{equation}
Critically, the functional form for $u_{\mathrm{Mi}}$ in Eq. \eqref{eq:u_Mi} fits numerical solutions $u$ of the original ODE system in Eqs. \eqref{eq:ODE_u} and \eqref{eq:ODE_v} during the Mi-phase strikingly well both in case of SIR-like dynamics (FIG. \ref{fig:Mi}A) as well as in case of dynamics reminiscent of logistic growth (FIG. \ref{fig:Mi}C). This demonstrates that our approximation of $u_{\mathrm{Mi}}(\tau)$ is flexible in the sense that it accounts for different types of dynamical behaviours. Moreover, it is the conversion function $f$ that determines which of these types of dynamics during the Mi-phase that the system undergoes.

%==========================================================
% Discussion
%==========================================================
\section{Discussion and conclusions}\label{sec:discussion}
By constructing and analysing a class of mechanistic models of prion-like replication, we suggest that the accumulation of toxic prions in mammalian cells over time is characterised by three phases referred to as HeMiTo-dynamics.
In the healthy phase, the concentration of healthy proteins normalises to physiological levels, while there is no formation of toxic protein.
Importantly, this phase is \textit{identical} for all types of conversion functions $f$ as Theorem \ref{thm:He} suggests, and the phase is dominated by the dynamics of healthy form. When interpreting the biological meaning of the He-phase, it is important to consider that we are assuming that the model parameters are in a regime where the TSS is stable. In this specific parameter regime, the toxic form will inevitably grow to a nonzero steady-state. Suppose the initial conditions are close to the physiological levels (the peak of the healthy concentration reached at the beginning of the Mi-phase). In that case, this phase is short-lived, and we should interpret the model as beginning at a time very close to disease initiation. In other words, the Mi-phase will begin abruptly. However, the He-phase may also represent the period in time in which the system parameters have changed, altering the levels of the healthy form prior to the production of the toxic form. In a biological setting, the heterodimer model may or may not exhibit a change in physiological levels of healthy protein before the production of toxic protein dependent on which model parameter induced the pathology (lowering of toxic clearance, higher production of healthy protein, or increased in the conversion of healthy to toxic form). As physiological levels of healthy protein remain mostly unchanged, future research should focus on mechanisms of disease initiation that do not alter healthy form concentration as predicted by the heterodimer model.

When a sufficiently high concentration of the healthy form is reached, the system enters the Mi-phase, where the healthy form reaches a maximum value while the toxic form increases exponentially. Here, the choice of conversion function $f$ determines the qualitative behaviour of the trajectories of the healthy form $u$ during the Mi-phase, and our asymptotic approximation $u_{\mathrm{Mi}}$ in Eq. \eqref{eq:u_Mi} indicates that there are two main types of dynamical behaviour during this phase. On the one hand, trajectories in the Mi-phase can behave like the SIR model where $u_{\mathrm{Mi}}(\tau)$ is a concave function which reaches a clear maximum point $u_{\max}$ before it decreases, and the reference model by Thompson et al.~\cite{thompson2020protein} defined by $f(v)=1.00$ has this type of SIR-like behaviour in the Mi-phase (FIG. \ref{fig:Mi}A). On the other hand, trajectories for the healthy form in the Mi-phase can also behave like a logistic growth function where the maximum value $u_{\max}$ corresponds to the carrying capacity of $u_{\mathrm{Mi}}(\tau)$, and, for instance, the model defined by $f(v)=0.57+v^{4}/(1+v^{5})$ has this type of dynamical behaviour (FIG. \ref{fig:Mi}C). Comparing these two cases, the formation of toxic prions, which occurs at the expense of the healthy form, is slower for the latter logistic growth type of dynamics compared to the SIR-like counterpart. This is also clear from the approximate system of ODEs in Eqs. \eqref{eq:ODE_U} and \eqref{eq:ODE_V} describing the dynamics of our prion models during the Mi-phase, as the conversion rate is essentially given by $f(v_{0})$, i.e. the value of the conversion function $f$ evaluated at the initial concentration of the toxic form $v_{0}$, which arises from the assumption that $f(v) \sim \mathcal{O}(v)$. In our simulations where $v_{0}=0.05$, it is clear that the value of $f(v_{0})$ is much lower for the function $f(v)=0.57+v^{4}/(1+v^{5})$ compared to the reference model $f(v)=1.00$ (FIG. \ref{fig:TSS}) and thus it is not surprising that the time to reach the TSS is much longer in the former case compared to the latter (FIG. \ref{fig:HeMiTo}). 

An interesting future line of research involves classifying prions in terms of specific conversion functions $f$. Using time series data of the concentration of toxic prions over time, our class of prion models defined by conversion functions $f$ can be used as a basis for model learning using neural networks. Ultimately, this would allow us to infer mechanistic models underlying experimental data of prion abundance over time, and potentially this can be used as a means to classify different prions in terms of conversion functions. Consequently, this work is a stepping stone towards improving our understanding of the fundamental workings of prions in the context of neurodegenerative diseases.

%%%%%%%%%%%%%%

%==========================================================
% Reproducibility
%=========================================================
\section*{Comment on reproducibility}
To run the analyses presented in this work, a script is available at the public github-repositry associated with this project; \url{https://github.com/JohannesBorgqvist/HeMiTo}.
%==========================================================
% Acknowledgements
%=========================================================
\section*{Acknowledgements}
The authors would like to thank Dr Alexander P. Browning and Dr Sam Palmer for reading the first draft of this manuscript. The authors would also like to thank Prof Philip Gerlee and Dr Adam Malik for fruitful discussions about perturbation analysis.

\section*{Funding}
JGB is funded by a grant from the Wenner-Gren foundations (Grant number: FT2023-0005).

\section*{CRediT author statement}
\textbf{JGB}: Conceptualization, Methodology, Visualization (made the figures), Writing - Original Draft, Writing - Review \& Editing, Formal analysis (derived the theorems and conducted all calculations), 
\textbf{CA}: Conceptualization, Writing - Review \& Editing.

%==========================================================
% References
%=========================================================

%\bibliographystyle{plain}
%\bibliography{reference}

%USE THE BELOW OPTIONS IN CASE YOU NEED AUTHOR YEAR FORMAT.
%\bibliographystyle{abbrvnat}
%\bibliography{reference}

\begin{thebibliography}{10}

\bibitem{brundin_prion-like_2010}
Patrik Brundin, Ronald Melki, and Ron Kopito.
\newblock Prion-like transmission of protein aggregates in neurodegenerative
  diseases.
\newblock {\em Nature Reviews Molecular Cell Biology}, 11(4):301--307, 2010.

\bibitem{cascarina_yeast_2014}
Sean~M. Cascarina and Eric~D. Ross.
\newblock Yeast prions and human prion-like proteins: sequence features and
  prediction methods.
\newblock {\em Cellular and Molecular Life Sciences}, 71(11):2047--2063, 2014.

\bibitem{condello2020prion}
Carlo Condello, William~F DeGrado, and Stanley~B Prusiner.
\newblock Prion biology: implications for alzheimer's disease therapeutics.
\newblock {\em The Lancet. Neurology}, 19(10):802--803, 2020.

\bibitem{dugger_pathology_2017}
Brittany~N. Dugger and Dennis~W. Dickson.
\newblock Pathology of neurodegenerative diseases.
\newblock {\em Cold Spring Harbor Perspectives in Biology}, 9(7):a028035, 2017.

\bibitem{fraser_prions_2014}
Paul~E. Fraser.
\newblock Prions and prion-like proteins.
\newblock {\em Journal of Biological Chemistry}, 289(29):19839--19840, 2014.

\bibitem{gerlee2022weak}
Philip Gerlee.
\newblock Weak selection and the separation of eco-evo time scales using
  perturbation analysis.
\newblock {\em Bulletin of Mathematical Biology}, 84(5):52, 2022.

\bibitem{greer_mathematical_2006}
Meredith~L. Greer, Laurent Pujo-Menjouet, and Glenn~F. Webb.
\newblock A mathematical analysis of the dynamics of prion proliferation.
\newblock {\em Journal of Theoretical Biology}, 242(3):598--606, 2006.

\bibitem{jack2013tracking}
Clifford~R Jack, David~S Knopman, William~J Jagust, Ronald~C Petersen,
  Michael~W Weiner, Paul~S Aisen, Leslie~M Shaw, Prashanthi Vemuri, Heather~J
  Wiste, Stephen~D Weigand, et~al.
\newblock Tracking pathophysiological processes in alzheimer's disease: an
  updated hypothetical model of dynamic biomarkers.
\newblock {\em The lancet neurology}, 12(2):207--216, 2013.

\bibitem{kermack1927contribution}
William~Ogilvy Kermack and Anderson~G McKendrick.
\newblock A contribution to the mathematical theory of epidemics.
\newblock {\em Proceedings of the royal society of london. Series A, Containing
  papers of a mathematical and physical character}, 115(772):700--721, 1927.

\bibitem{lemarre2020unifying}
Paul Lemarre, Laurent Pujo-Menjouet, and Suzanne~S Sindi.
\newblock A unifying model for the propagation of prion proteins in yeast
  brings insight into the [psi+] prion.
\newblock {\em PLoS computational biology}, 16(5):e1007647, 2020.

\bibitem{meisl_molecular_2020}
Georg Meisl, Tuomas~{PJ} Knowles, and David Klenerman.
\newblock The molecular processes underpinning prion-like spreading and seed
  amplification in protein aggregation.
\newblock {\em Current Opinion in Neurobiology}, 61:58--64, 2020.

\bibitem{prodanov2023asymptotic}
Dimiter Prodanov.
\newblock Asymptotic analysis of the sir model and the gompertz distribution.
\newblock {\em Journal of Computational and Applied Mathematics}, 422:114901,
  2023.

\bibitem{prusiner_molecular_1991}
Stanley~B. Prusiner.
\newblock Molecular biology of prion diseases.
\newblock {\em Science}, 252(5012):1515--1522, 1991.

\bibitem{Sindi17}
Suzanne~S. Sindi.
\newblock Mathematical modeling of prion disease.
\newblock In Yusuf Tutar, editor, {\em Prion}, chapter~10. IntechOpen, Rijeka,
  2017.

\bibitem{sindi2009prion}
Suzanne~S Sindi and Tricia~R Serio.
\newblock Prion dynamics and the quest for the genetic determinant in
  protein-only inheritance.
\newblock {\em Current opinion in microbiology}, 12(6):623--630, 2009.

\bibitem{stone-weierstrass}
M.~H. Stone.
\newblock The generalized weierstrass approximation theorem.
\newblock {\em Mathematics Magazine}, 21(4):167--184, 1948.

\bibitem{thompson_role_2021}
T.~B. Thompson, G.~Meisl, T.~P.~J. Knowles, and A.~Goriely.
\newblock The role of clearance mechanisms in the kinetics of pathological
  protein aggregation involved in neurodegenerative diseases.
\newblock {\em The Journal of Chemical Physics}, 154(12):125101, 2021.

\bibitem{thompson2020protein}
Travis~B Thompson, Pavanjit Chaggar, Ellen Kuhl, Alain Goriely, and
  Alzheimer’s Disease~Neuroimaging Initiative.
\newblock Protein-protein interactions in neurodegenerative diseases: A
  conspiracy theory.
\newblock {\em PLoS computational biology}, 16(10):e1008267, 2020.

\bibitem{verhulst1838logistic}
Pierre~François Verhulst.
\newblock Notice sur la loi que la population poursuit dans son accroissemen.
\newblock {\em Correspondance Math\'{e}matique et Physique}, X, 1838.

\bibitem{WEICKENMEIER2019_prion-model}
Johannes Weickenmeier, Mathias Jucker, Alain Goriely, and Ellen Kuhl.
\newblock A physics-based model explains the prion-like features of
  neurodegeneration in alzheimer’s disease, parkinson’s disease, and
  amyotrophic lateral sclerosis.
\newblock {\em Journal of the Mechanics and Physics of Solids}, 124:264--281,
  2019.

\end{thebibliography}

%==========================================================
% Appendices
%=========================================================
\begin{appendices}
\begin{center}
    \Large \textbf{Appendices}\\
\end{center}
%===========================================================
% Non-dimensionalisation
%===========================================================
\section{Details behind the non-dimensionalisation of the class of prion models}\label{sec:non_dim}
We assume that the conversion function is dimensionless, and essentially we have that $f(T)=f(v)$. Multiplying the first ODE for $H$ in Eq. \eqref{eq:ODE_H} by $k_{3}/k_{4}$, we get
\begin{align*}
    \underset{=\mathrm{d}u/\mathrm{d}t}{\underbrace{\left(\frac{k_{3}}{k_{4}}\right)\dot{H}}}&=\left(\frac{k_{1}k_{3}}{k_{4}}\right)-k_{2}\underset{=u}{\underbrace{\left[\left(\frac{k_{3}}{k_{4}}\right)H\right]}}-\frac{k_{3}k_{4}}{k_{3}}\underset{=u}{\underbrace{\left[\left(\frac{k_{3}}{k_{4}}\right)H\right]}}\underset{=v}{\underbrace{\left[\left(\frac{k_{3}}{k_{4}}\right)T\right]}}\underset{=f(v)}{\underbrace{f(T)}}\\\Longrightarrow\frac{\mathrm{d}u}{\mathrm{d}t}&=\left(\frac{k_{1}k_{3}}{k_{4}}\right)-k_{2}u-k_{4}uvf(v)\,.
\end{align*}
Similarly, multiplying the second ODE for $T$ in Eq. \eqref{eq:ODE_T} by $k_{3}/k_{4}$, we get
\begin{align*}
    \underset{=\mathrm{d}v/\mathrm{d}t}{\underbrace{\left(\frac{k_{3}}{k_{4}}\right)\dot{T}}}&=\left(\frac{k_{3}k_{4}}{k_{3}}\right)\underset{=u}{\underbrace{\left[\left(\frac{k_{3}}{k_{4}}\right)H\right]}}\underset{=v}{\underbrace{\left[\left(\frac{k_{3}}{k_{4}}\right)T\right]}}\underset{=f(v)}{\underbrace{f(T)}}-k_{4}\underset{=v}{\underbrace{\left[\left(\frac{k_{3}}{k_{4}}\right)T\right]}}\\
    \Longrightarrow\frac{\mathrm{d}v}{\mathrm{d}t}&=k_{4}uvf(v)-k_{4}v\,.    
\end{align*}
Multiplying the ODE for $\mathrm{d}u/\mathrm{d}t$ by $k_{2}/(k_{1}k_{3}f(T_{2}^{\star}))$ yields
\begin{align*}
    \left(\frac{k_{2}}{k_{1}k_{3}f(T_{2}^{\star})}\right)\frac{\mathrm{d}u}{\mathrm{d}t}&=\left(\frac{k_{2}}{k_{1}k_{3}f(T_{2}^{\star})}\right)\left(\frac{k_{1}k_{3}}{k_{4}}\right)-\left(\frac{k_{2}}{k_{1}k_{3}f(T_{2}^{\star})}\right)k_{2}u-\left(\frac{k_{2}}{k_{1}k_{3}f(T_{2}^{\star})}\right)k_{4}uvf(v)\\
    \Longrightarrow\frac{\mathrm{d}u}{\mathrm{d}\tau}&=c_{1}-c_{2}u-\varepsilon{u}vf(v)\\
\end{align*}
where the parameters $c_{1}$ and $c_{2}$ are defined in Eq. \eqref{eq:parameters} and the perturbation parameter $\varepsilon$ is defined in Eq. \eqref{eq:eps}. Lastly, multiplying the ODE for $\mathrm{d}v/\mathrm{d}t$ by $k_{2}/(k_{1}k_{3}f(T_{2}^{\star}))$ yields
\begin{align*}
    \left(\frac{k_{2}}{k_{1}k_{3}f(T_{2}^{\star})}\right)\frac{\mathrm{d}v}{\mathrm{d}t}&=\left(\frac{k_{2}}{k_{1}k_{3}f(T_{2}^{\star})}\right)k_{4}uvf(v)-\left(\frac{k_{2}}{k_{1}k_{3}f(T_{2}^{\star})}\right)k_{4}v\Longrightarrow\frac{\mathrm{d}v}{\mathrm{d}\tau}=\varepsilon\left(uvf(v)-v\right)\,,\\
\end{align*}
which is the desired result.
%===========================================================
% Existence of SSs
%===========================================================
\section{Proof of Theorem \ref{thm:SS}: existence of steady states}\label{sec:SS}
\begin{proof}
The steady states $(u^{\star},v^{\star})$ solve
\begin{align}
    0&=c_{1}-c_{2}u^{\star}-\varepsilon{u^{\star}}v^{\star}f(v^{\star})\,,\label{eq:SS_1}\\
    0&=\varepsilon{v^{\star}}\left({u^{\star}}f(v^{\star})-1\right)\,.\label{eq:SS_2_temp}
\end{align}
A solution of Eq. \eqref{eq:SS_2_temp} is given by $v_{1}^{\star}=0$, and substituting this into Eq. \eqref{eq:SS_1} and solving for $u_{1}^{\star}$ yields the HSS in Eq. \eqref{eq:HSS}. Another solution of Eq. \eqref{eq:SS_2_temp} is given by $u_{2}^{\star}=\left(f(v_{2}^{\star})\right)^{-1}$ and substituting this into Eq. \eqref{eq:SS_2_temp} and solving for $v_{2}^{\star}$ yields the TSS in Eq. \eqref{eq:TSS}. 
\end{proof}
%===========================================================
% HSS saddle
%===========================================================
\section{Proof of Proposition \ref{thm:saddle_HSS}: condition for the HSS being a saddle point}\label{sec:saddle_HSS}
\begin{proof}
    The Jacobian matrix $\mathcal{J}(u,v)$ is given by
    \begin{equation}
    \mathcal{J}(u,v)=\begin{pmatrix}-c_{2}-\varepsilon{v}f(v) & -\varepsilon{u}(f(v)+vf'(v))\\\varepsilon{v}f(v) & \varepsilon{u}(f(v)+vf'(v))-\varepsilon\end{pmatrix}\,,
        \label{eq:Jacobian}
    \end{equation}
    and its determinant by
    \begin{align}
        \mathrm{Det}(\mathcal{J}(u,v))&=\varepsilon(c_{2}+\varepsilon{v}f(v))-c_{2}\varepsilon{u}(f(v)+vf'(v))\,.\label{eq:det}   
    \end{align}
    Evaluating the determinant at the HSS in Eq. \eqref{eq:HSS} yields
\begin{align}
        \mathrm{Det}\left(\mathcal{J}\left(\frac{c_{1}}{c_{2}},0\right)\right)&=c_{2}\varepsilon\left(1-\frac{c_{1}}{c_{2}}f(0)\right)\,,\label{eq:det_HSS}   
    \end{align}    
    and the HSS is a saddle point when this determinant is negative which corresponds to the parameter condition in Eq. \eqref{eq:HSS_saddle}.
\end{proof}

%===========================================================
% TSS stable node
%===========================================================
\section{Proof of Theorem \ref{thm:stable_TSS}: conditions defining prion toxicity}\label{sec:stable_TSS}
\begin{proof}
Evaluating the determinant at the TSS in Eq. \eqref{eq:TSS} yields
\begin{align}
    \mathrm{Tr}\left(\mathcal{J}\left(u_{2}^{\star},v_{2}^{\star}\right)\right)=\mathrm{Tr}\left(\mathcal{J}\left(\frac{1}{f(v_{2}^{\star})},v_{2}^{\star}\right)\right)&=-c_{2}+\varepsilon{v}_{2}^{\star}\left(f'(v_{2}^{\star})-f(v_{2}^{\star})\right)\,,\label{eq:trace_TSS} \\
        \mathrm{Det}\left(\mathcal{J}\left(u_{2}^{\star},v_{2}^{\star}\right)\right)=\mathrm{Det}\left(\mathcal{J}\left(\frac{1}{f(v_{2}^{\star})},v_{2}^{\star}\right)\right)&=\varepsilon{v}_{2}^{\star}\left(f(v_{2}^{\star})-\frac{c_{2}}{f(v_{2}^{\star})}f'(v_{2}^{\star})\right)\,.\label{eq:det_TSS}  
    \end{align}  
    The TSS is a stable node if the trace in Eq. \eqref{eq:trace_TSS} is negative and the determinant in Eq.  \eqref{eq:det_TSS} is positive. These two requirements can be expressed in terms of the following inequalities
    \begin{align}
        \underset{>0}{\underbrace{\frac{c_{2}}{\varepsilon{v}_{2}^{\star}}+f(v_{2}^{\star})}}&>f'(v_{2}^{\star})\,,\label{eq:trace_TSS_2}\\
        \underset{>0}{\underbrace{\frac{1}{c_{2}}f(v_{2}^{\star})^{2}}}&>f'(v_{2}^{\star})\,,\label{eq:det_TSS_2}
    \end{align}
    and clearly both of these are satisfied whenever $f'(v_{2}^{\star})\leq{0}$.
\end{proof}

%===========================================================
% He-phase
%===========================================================
\section{Proof of Theorem \ref{thm:He}: approximate analytical solutions in the He-phase}\label{sec:He}
\begin{proof}
Substituting the perturbation ans\"{a}tze in Eqs. \eqref{eq:u_approx} and \eqref{eq:v_approx} into the system of ODEs in Eqs. \eqref{eq:ODE_u} and \eqref{eq:ODE_v} subject to the initial conditions $(u_{0},v_{0})$ in Eq. \eqref{eq:IC_u_and_v} and then extracting the $\mathcal{O}(1)$ terms yields the following system for the outer solutions
\begin{align}
    \dot{u}_{\mathrm{He}}&=c_{1}-c_{2}u_{\mathrm{He}}\,,\quad&{u}_{\mathrm{He}}(\tau=0)=u_{0}\,,\label{eq:ODE_u_He}\\
    \dot{v}_{\mathrm{He}}&=0\,,\quad&{v}_{\mathrm{He}}(\tau=0)=v_{0}\,.\label{eq:ODE_v_He}
\end{align}
The solutions of these equations are given by Eqs. \eqref{eq:u_He} and \eqref{eq:v_He}, respectively.
\end{proof}

\end{appendices}

%--------------------------------------------------------------------------------------
% THE DOCUMENT ENDS
%--------------------------------------------------------------------------------------
\end{document}